\documentclass[a4paper,12pt]{article}
\usepackage[latin1]{inputenc}
\usepackage[english]{babel}
\usepackage{amssymb}
\usepackage{makeidx}
\usepackage{graphicx}
\usepackage{mathrsfs}
\usepackage{dsfont}

\usepackage{color}
\usepackage{ulem}
\usepackage{amsfonts}
\usepackage{amsmath}
\allowdisplaybreaks
\usepackage{times}
\usepackage{amsthm}
\usepackage{epsfig}
\newcommand{\DATUM}{}              
\newcommand{\change}
{{\marginpar{\#}}}        
\newcommand{\vphi}{{\varphi}}           

\newcommand{\la}{\langle}
\newcommand{\ra}{\rangle}

\newcommand{\one}{{\bf 1}}

\newcommand{\cB}{{\mathcal{B}}}

\newcommand{\cF}{{\mathcal{F}}}

\newcommand{\cH}{{\mathcal{H}}}

\newcommand{\cN}{{\mathcal{N}}}         

\newcommand{\cS}{{\mathcal{S}}}

\newcommand{\NN}{\mathbb{N}}            
\newcommand{\CC}{\mathbb{C}}            
\newcommand{\cirS}{\mathop{\bigcirc\kern -.73em {\scriptstyle{\rm S}}}}

\newcommand{\Tr}{{\rm Tr}}

\newcommand{\at}{{at}}
\newcommand{\rad}{{\mathrm{ph}}}
\newcommand{\gs}{{\mathrm{gs}}}

\newcommand{\gap}{{\mathrm{gap}}}

\newcommand{\secret}[1]{}

\renewcommand{\thesection}
{\arabic{section}}                      
\renewcommand{\theequation}
{\thesection.\arabic{equation}}        
\newtheorem{theorem}{Theorem}[section]         
\newtheorem{lemma}[theorem]{Lemma}             
\newtheorem{corollary}[theorem]{Corollary}     
\newtheorem{remark}[theorem]{Remark}           
\newtheorem{proposition}[theorem]{Proposition} 
\theoremstyle{plain}
\def\beq{\begin{equation}}
\def\ene{\end{equation}}

\begin{document}
\bibliographystyle{plain}

\title{Existence of Ground State Eigenvalues for the Spin-Boson Model with Critical Infrared Divergence and Multiscale Analysis}

\author{
Volker Bach \\
\small{Institut f\"ur Analysis und Algebra}\\
\small{Technische Universit\"at Braunschweig } \\[-1ex]
\small{Germany (v.bach@tu-bs.de)}
\and
  Miguel Ballesteros    \\
\small{ Instituto de Investigaciones en Matem\'aticas Aplicadas y en Sistemas (IIMAS).}\\
\small{  Universidad Nacional Aut\'onoma de M\'exico (UNAM)}\\[-1ex] 
\small{  Mexico (miguel.ballesteros@iimas.unam.mx)} 
\and 
Martin K\"onenberg  \\
\small{Institut f\"ur Analysis, Dynamik und Modellierung}\\
\small{Universit\"at Stuttgart }\\[-1ex]
\small{ Germany (martin.koenenberg@mathematik.uni-stuttgart.de)} 
\and 
Lars Menrath\\
\small{Institut f\"ur Analysis und Algebra}\\
\small{Technische Universit\"at Braunschweig } \\[-1ex]
\small{Germany (l.menrath@tu-bs.de)}
}
\date{\DATUM}

\maketitle




\begin{abstract} {\color{black} \secret{We study} A} two-level atom
  coupled to the radiation field {\color{black} is studied. First
    principles in physics suggest that the\secret{The}} coupling
  function, representing the interaction between the atom and the
  radiation field, behaves {\color{black}\secret{as $|k|^{-1/2}$ -- as
      $k$ tends to zero -- being $k$ the photon momentum} like
    $|k|^{-1/2}$, as the photon momentum $k$ tends to zero.\secret{This
      is the asymptotic behavior that can be derived from first
      principles in physics.}}  Previous results on non-existence of
  ground state eigenvalues suggest that in the most general case
  binding does not occur in the spin-boson model, i.e., the minimal
  energy of the atom-photon system is not an eigenvalue of the energy
  operator. {\color{black}\secret{However, further hypotheses on the
      coupling function allow proving existence of ground state
      eigenvalues.} Hasler and Herbst have shown
    \cite{HaslerHerbst2011}, however, that under the additional
    hypothesis that the coupling function be off-diagonal -which is
    customary to assume- binding does indeed occur.} In this paper
  {\color{black}\secret{we prove that, if} an alternative proof of
    binding in case of off-diagonal coupling is given, i.e., it is
    proven that,} if the coupling function is off-diagonal,
  {\color{black}\secret{ground state eigenvalues exist for} the ground
    state energy of} the spin-boson model {\color{black} is an
    eigenvalue of the Hamiltonian}. We develop a multiscale method
  that can be applied in the situation we study, identifying a new key
  symmetry operator which we use to demonstrate that the most singular
  terms appearing in the multiscale analysis vanish.
 \end{abstract}

\section{Introduction}

The precise description of nonrelativistic matter in interaction with
the quantized radiation field has been in the focus of mathematical
research ever since the proposal of the quantization of the radiation
field by Dirac more than eighty years ago \cite{Dirac1930}. 

The invention of the Laser some fifty years ago necessitated the
development of a simplified, yet, adequate model for the description
of its mechanism in theoretical physics. It proved useful to simplify
the model of matter from atom and molecules to two-level atoms. The
corresponding model, known as the \textit{spin-boson model}, became
the work horse of quantum optics and is nowadays of key importance for
quantum computing, with the interpretation of the two-level atom as a
qubit.

Starting more than twenty years ago, the mathematical aspects of the
models of nonrelativistic matter coupled to the quantized radiation
field -known as \textit{nonrelativistic quantum electrodynamics, NR
  QED}- were systematically investigated. In contrast to many models
from relativistic quantum mechanics or quantum field theory, the
models in NR QED are defined by a self-adjoint, semi-bounded
Hamiltonian $H = H^* \geq c > - \infty$ acting on the tensor product
$\cH = \cH_\at \otimes \cF$ of the Hilbert spaces $\cH_\at$ of
matter and $\cF$ of the radiation field, respectively. During the
past two decades or so, for many models of NR QED, basic spectral
properties like \textit{binding} and the \textit{existence of
  resonances} have been established. These represent the expected fate
of the eigenvalues of the atom as it is coupled to the radiation
field: The lowest spectral point persists to be an eigenvalue and all
other atomic eigenvalues are unstable and give rise to metastable
states, the resonances.

Specifically, \textit{binding} means that the infimum $E_\gs :=
\inf\sigma(H) > - \infty$ of the spectrum of the Hamiltonian is an
eigenvalue, called the ground state energy, with an eigenvector
$\vphi_\gs \in \cH$, called the ground state, i.e., $H \vphi_\gs =
E_\gs \vphi_\gs$. 

Binding in NR QED was established for atoms and molecules coupled to
the radiation field
\cite{BachFroehlichSigal1995,BachFroehlichSigal1998a}, as well as, for
the spin-boson model \cite{AraiHirokawa1997} about twenty years ago
under the assumption that the coupling function $G(k)$ is slightly
more regular, $|G(k)| \leq C |k|^{-\frac{1}{2} + \mu}$, for some $C <
\infty$ and $\mu >0$, in the infrared limit $k \to 0$, than what
derives from first principles in physics, namely, $|G(k)| \sim C
|k|^{-\frac{1}{2}}$, as $k \to 0$.

For these {\color{black} latter,} more singular models, with $|G(k)|
\sim C |k|^{-\frac{1}{2}}$, as $k \to 0$, binding was shown to hold
true a few years later \cite{BachFroehlichSigal1999} in the special,
but physically most relevant case that the radiation field is
minimally coupled to the electrons of the atom. Here, it was used that
the model possesses additional symmetries such as the $U(1)$-gauge
symmetry. The key identity (in the case of one electron, as for the
hydrogen atom) made use of in the proof is $\vec{v} = i[ H, \vec{x}
]$, where $\vec{v} = -i\vec{\nabla}_x - \vec{A}(\vec{x})$ is the
velocity operator and $\vec{x}$ the position operator of the electron.

Following an argument of Fr\"ohlich \cite{Froehlich1974} it was
assumed for many years \cite{AraiHirokawaHiroshima1999} that the
spin-boson model with singular coupling does not bind in the above
sense, but rather possesses a ground state that is revealed by a
(non-unitary) change of the representation of the canonical
commutation relations. In view of this common belief the recent proof
of Hasler and Herbst \cite{HaslerHerbst2011} for binding of the
spin-boson model with singular coupling is a remarkable result. Their
proof uses the renormalization group based on the isospectral
Feshbach-Schur map developed in
\cite{BachFroehlichSigal1998a,BachFroehlichSigal1998b,BachChenFroehlichSigal2003}. Their
additional key observation is that {\color{black} since} there is no
self-interaction of each of the two levels of the atom, but only a
coupling to one another, {\color{black}\secret{and that this fact induces
    a regularization of}} the (discrete) flow equation defined by the
renormalization group {\color{black} is more regular than it seems to be
  at first glance}.

In the present paper we give an alternative proof for binding of the
spin-boson model with singular coupling. We consider the spin-boson 
Hamiltonian 
\begin{align} \label{eq-intro-1}
 H  
\ := \
H_\at + H_\rad + \Phi(G) ,
\end{align}
where $H_\rad \equiv \one_\at \otimes H_\rad$ is the field Hamiltonian
and $H_\at = \sigma_3 + \one_\at \equiv H_\at \otimes \one_{\cF}$ is
the Hamiltonian of the two-level atom, with $\sigma_\nu$ denoting the
Pauli matrices. Furthermore, $\Phi(G)$ is the interaction with field
operator $\Phi(G) = a^*(G) + a(G)$, with $G \equiv g \sigma_1 \otimes
h(k)$, where $h$ is a compactly supported coupling function obeying
$|h(k)| \sim c |k|^{-\frac{1}{2}}$, as $k \to 0$, and $g \geq 0$ is the
coupling strength, see Eqs.~\eqref{eq.7}--\eqref{eq.5}. For this
Hamiltonian $H$ we prove our main result, Theorem~\ref{ESE},
which states that the infimum of its spectrum is an eigenvalue.

Our construction is based on \textit{Pizzo's method} \cite{Pizzo2003}, rather than the
renormalization group induced by the Feshbach-Schur map. That is, we
consider a sequence $H_n \equiv H( G_n)$ of regularized Hamiltonians
whose coupling functions $G_n(k) = \one( |k| \geq \rho_n) G(k)$ are the
restrictions of $G$ to photon momenta larger than $\rho_n = \kappa \,
\gamma^n$, for some fixed $\gamma <1$ and all $n \in \NN$. Following the
idea originally formulated by Pizzo, we inductively prove that each
$H_n$ shows binding with a ground state energy $E_n$ being a
non-degenerate eigenvalue with normalized eigenvector $\phi_n$ and
rank-one eigenprojection $P_n = |\phi_n\ra\la\phi_n|$.

It is fairly easy to establish the existence of these eigenprojections
$P_n$, for each $n$, and the principal difficulty of this and all
other such constructions lies in the proof of convergence $P_n \to
P_\gs $ of $P_n$ (here the range of $P_{\rm gs}$ consists of ground state eigenvectors of $H$). The additional property from which we derive
this convergence in this seemingly too singular case is that the
original Hamiltonian $H$, as well as, all Hamiltonian operators $H_n$,
commute with a symmetry {\color{black} $\cS = \sigma_3  (-\one)^{\cN_\rad}$, where
$\cN_\rad$} is the photon number operator. This symmetry induces a decomposition
of the Hilbert space into the two subspaces $\cH_\pm$ corresponding to
the eigenvalues $\pm 1$ of $\cS$. The operators $H$ and $H_n$ leave
these subspaces invariant, and from this we draw the consequence that
\begin{align} \label{eq-intro-2}
\Tr\{ P_n \, \sigma_1 \, P_n \} \ = \ 0 ,
\end{align}
for all $n \in \NN$,  which enters our proofs at key steps.

\subsection{The Model}
We study a two-level atom interacting with the radiation field. We assume, without loss of generality, that  the ground state (free) energy of the atom equals $0$ and the excited energy equals $2$. In this paper we only consider
non-degenerate energies. Therefore, the Hamiltonian of the atom alone is given by the matrix
\begin{align}
 H_{ \rm at}:=
 \begin{pmatrix}
 2 & 0 \\
 0 & 0
 \end{pmatrix}  \label{eq.6}
\end{align}
acting on the atom Hilbert space 
\begin{equation}
\mathcal{H}_{ \rm at} : = \mathbb{C}^2.
\end{equation}
For every Hilbert space $\mathfrak{h}$, we denote by 
\begin{align}
\mathcal{F}[\mathfrak{h}] := \mathbb{C} \oplus \bigoplus_{\ell=1}^{\infty} \mathcal{S}_\ell \bigotimes_{k=1}^{\ell} \mathfrak{h} \,  \label{eq.2}
\end{align}
its associated bosonic (symmetric) Fock space. $\Omega_{\mathfrak{h}} \in \mathcal{F}[\mathfrak{h}]$ denotes the vacuum vector. Here $\mathcal{S}_{\ell}$ denotes the orthogonal projection onto the subspace of totally symmetric tensors.  The Hilbert space for the radiation field is defined to be  
\begin{align}
\mathcal{F} \equiv  \mathcal{F}[L^{2}[\mathbb{R}^3]] \, .  \label{eq.2prima}
\end{align}
Note that it is not quite adequate to call the quanta of this scalar field {${\it photons} $, as polarization is not taken into account.  The free photon energy is given by the operator 
\begin{align}
 H_\rad \equiv  H_\rad(\omega) :=  d\Gamma(\omega) = \int_{\mathbb{R}^3} \omega(k) a^{*}(k) a(k) dk \, ,  \label{eq.3}
\end{align}
where $\omega(k):= \vert k \vert$, and $a^{*}(k)$, $a(k)$ denote the creation and annihilation operators representing the canonical commutation relations on $\mathcal{F}$, i.e.,
\begin{align}
\big[a(k), a^{*}(k') \big] =\delta(k-k') \, , \, \big[a(k), a(k') \big] =\big[a^{*}(k), a^{*}(k') \big]=0 \, , \, a(k)\Omega=0 \, ,  \label{eq.4}
\end{align}
for all $k, k' \in \mathbb{R}^3$, in the sense of operator-valued distributions. 
In Eq.~\eqref{eq.3} we use Nelson's notation for the second quantization $\Gamma(A)$ of a one-photon operator $A$. We furthermore introduce the photon number operator $ \mathcal{N}_\rad  $, defined on 
$\mathcal{H}$, by the following equation
\begin{align}\label{numberfull}
\mathcal{N}_\rad : =  H_\rad(  {\bf 1}_ { \mathbb{R}^3})
\equiv d \Gamma (   {\bf 1}_ { \mathbb{R}^3  }  ).
\end{align}
The Hilbert space of the (full) atom-photon system is the tensor product of the atom and the photon Hilbert spaces:
\begin{align}
\mathcal{H} := \mathcal{H}_{\rm at} \otimes \mathcal{F}. \label{eq.1}
\end{align}
The interaction between the atom and the photon field is expressed in terms of the field operator,
\begin{align} 
\Phi (G) := 
\int_{\mathbb{R}^3} [G(k) \otimes a^{*}(k) + G^*(k) \otimes a(k) ] dk , \label{eq.7}
\end{align}
where we assume that $G$ is of the form
\begin{align}
G(k) :=  g \frac{\Lambda(k)}{ 4 \pi \sqrt{\omega(k)}} f(k)  \sigma_1
 \, , \forall k \in \mathbb{R}^3 \, , \label{eq.7prima}
\end{align}
with $\Lambda(k):= {\bf 1}_{ \{   k : |k| < \kappa \}  } $ (the characteristic function of the set $  \{   k : |k| < \kappa \}$) being an ultraviolet cutoff and the coupling constant $ g >0 $ being a small parameter. 
 For convenience (without loss of generality) we choose the UV-cutoff scale as $\kappa   < 1$.  We assume that $f = f^* \in L^{\infty}(\mathbb{R}^3) $ is {\color{black} uniformly} bounded by $1$ {\color{black}\secret{(with respect to the $L^\infty $--norm)}} and that $\sigma_1$ is the first Pauli matrix (the diagonal entries equal zero and the other entries equal $1$). The energy of the full system is {\color{black}\secret{, of course,}} the sum of all energies just introduced,
\begin{align}
H := H_{\rm at}  + \Phi (G) + H_\rad \, . \label{eq.5}
\end{align}
Here we use the identifications $   H_{\rm at} \equiv  H_{\rm at} \otimes \mathds{1}_{\mathcal{F}} $, $ H_\rad \equiv \mathds{1}_{\mathcal{H}_{\rm at}} \otimes H_\rad. $   In general, for pairs of Hilbert spaces $V_1$ and $V_2$ and operators $ A_1  $  and $A_2$ defined on $ V_1 $ and $V_2$, respectively, we leave out trivial tensor factors and write
\begin{equation}
A_1 \equiv A_1 \otimes \mathds{1}_{V_2} , \hspace{2cm} A_2 \equiv  \mathds{1}_{V_1}\otimes A_2.  
\end{equation}

\subsection{Main Theorem and Outline of its Proof}
Our main result is proven in Section \ref{Main-Section}, specifically it is restated in Theorems \ref{MAIN} and \ref{GS-Projection} (see also Remark \ref{Remark}). Here we provide the core of our results in the next  
\begin{theorem} \label{ESE}
For sufficiently small {\color{black} $g > 0$} the bottom of the spectrum, 
\begin{equation}
E_{\rm gs} : = \inf \sigma(H),
\end{equation}
is an eigenvalue of $H$. 
\end{theorem}
 


The proof of Theorem \ref{ESE} uses perturbation theory in a {\color{black}\secret{clever} non-trivial}} way. Notice that the free Hamiltonian 
\begin{equation}
H_{\rm Free} : = H_{\rm at}  +  H_\rad
\end{equation}
has zero as an eigenvalue at the bottom of its spectrum. As  the spectrum of $H_{\rm Free}$ is $[0, \infty)$, $0 $ is immersed in the continuum. Thus, standard perturbation theory of isolated eigenvalues of finite multiplicity cannot be applied and multiscale or renormalization techniques must be used, we utilize multiscale analysis. Since the coupling function behaves asymptotically as $\| G(k) \| \sim |k|^{-1/2}$, the interaction $\Phi(G)$ scales like the field Hamiltonian $H_\rad$ under unitary dilations, namely, like an inverse length. Consequently, $\Phi(G)$ is a {\it marginal} perturbation of $H_\rad$, which makes the direct application of renormalization group schemes difficult. In order to prove that $\Phi(G)$ is actually {\it marginally irrelevant}, we identify a new symmetry $\mathcal{S}$ of the system which allows us to conclude that the matrix element $\langle \psi | \sigma_1 \psi \rangle $ vanishes, for any eigenvector of $H$. One of the main purposes of this paper is to demonstrate that this information can be used to show the convergence of the ground state construction.       

The multiscale analysis is based {\color{black}\secret{in} on} the
construction of {\color{black}\secret{an infrared regular sequence of} a
  sequence of infrared regular} Hamiltonians whose ground state
projections converge to a projection with range consisting of
eigenvectors of $H$. The elements of {\color{black}\secret{the} this}
sequence of Hamiltonians cut off {\color{black}\secret{low energies}
  small momenta}, but progressively incorporate {\color{black}
 \secret{more low energies, as the parameter of the sequence goes,}
  ever smaller momenta} in such a way that {\color{black}\secret{an
    infinite iteration of the scheme brings all energies into the
    picture} eventually all momenta are taken into account}. More
specifically, we proceed as follows:

The infrared cutoff functions are characterized by a decreasing sequence $(\rho_n)_{n \in \mathbb{N}_0}$ of numbers
\begin{align}
\rho_n := \kappa \gamma^n < 1 \, , \, \forall n \in \mathbb{N}_0 \, ,  
\end{align}
for some specifically small parameter $\gamma \in (0, 1)$ that will be conveniently chosen later on (recall that we set $\kappa <1$ above). Here  $\mathbb{N}_0 : = \mathbb{N}\cup \{ 0 \}$.  
The number $\rho_n$ represents the lowest allowed photon energy at step $n$.   
We cut off energies below $\rho_n$ in the following manner: Define
\begin{align}\label{mite}
\omega_n := &  {\bf 1}_ { \mathbb{R}^3 \backslash B_n  } \, \omega  \, ,  & G_n:= &  {\bf 1}_ { \mathbb{R}^3 \backslash B_n } \, G \, , 
\end{align}
where 
{\color{black} 
\begin{align}\label{B_n}
B_n := &  \big\{ k \in \mathbb{R}^3 : \quad |k| < \rho_n \big\}
\end{align}
}is the ball centered at the origin with radius $\rho_n$.  In
Eq.~\eqref{mite}, $\omega_n $ cuts off the free photon energies below
$\rho_n$. Similarly $ G_n$ cuts off interacting energies. The symbol $
{\bf 1}_{A} $ represents the characteristic (or indicator) function of
the set $A$. Notice that $G_0 = 0$, since $G$ is supported in
$B_0$. {\color{black}\secret{We now} Now we} define a sequence of
{\color{black} infrared-cutoff} Hamiltonians $(H_n)_{n \in
  \mathbb{N}}$. Set
\begin{align}
H_{n} := H_{at}  +  \Phi (G_n) + H_{ph}(\omega_n) \, 
\end{align}
on $\mathcal{H}_{n} := \mathcal{H}_{at}  \otimes \mathcal{F}_n $, where  
\begin{align}
\mathcal{F}_n := \mathcal{F}[L^{2}[\mathbb{R}^3 \backslash B_{n} ]] \, , \label{eq.2.2primas} 
\end{align} 
with vacuum state denoted by $\Omega_n$. Cutting {\color{black} off the
  photon} energies below $\rho_n$ {\color{black}\secret{off}} implies
that{\color{black}, for every $n$,} the Hamiltonian $H_n$ {\color{black}\secret{, for every $n$,}} has an isolated eigenvalue,
\begin{equation}\label{En}
E_n : = \inf \, \sigma(H_n),
\end{equation}
at the bottom of its spectrum{\color{black},} and we prove the gap above
the spectrum {\color{black}\secret{is bigger or equal than} to be bigger
  than, or equal to, $\frac{1}{2}\rho_n$}. The idea of our construction
is quite natural, we prove that the sequence of eigenvalues converges
and that the limit of it is actually an eigenvalue of the spectrum of
$ H $, namely, its ground state energy.


Besides the considerations involving the symmetry $\mathcal{S}$, we
need some robust estimates which are standard, but included in this
paper in Section \ref{regular}, for {\color{black}\secret{the}}
convenience of the reader.  The principal properties we prove are
bounds for the energy differences $ |E_{n + 1} - E_{n}| $ and the
distance {\color{black} ${\rm gap}_n > 0$} of $E_n$ to the rest of the
spectrum of $H_n$, {\color{black}\secret{we call it the gap at step $n$
    (${\rm gap}_n$)} which we call the gap at step $n$}, for every
$n$: In Proposition \ref{energies} we {\color{black}\secret{prove} show}
that
\begin{equation}\label{C1}
| E_{n + 1} - E_n | <   g \rho_n, 
\end{equation} 
and in Lemma \ref{gapst} we prove that, 
for every $n \in \mathbb{N}_0$,
\begin{equation}\label{int1}
 \gap_{n} :=  \inf \big\{ \sigma( H_{n} ) \backslash \{ E_{n} \} \big\} \geq  \frac{1}{2}  \rho_{n},
\end{equation}
for sufficiently small $g$, uniformly {\color{black}\secret{with respect
    to} in} $n$. Eq.~\eqref{C1} already implies the convergence of the
sequence $\{ E_n \}_{n \in \mathbb N}$. We actually prove (see Remark
\ref{Remark}) that
\begin{equation}
E_{\rm gs}  = \lim_{n\to \infty} E_n,
\end{equation}
{\color{black} where we} recall that $E_{\rm gs} $ is the infimum of the spectrum of $H$. 
The proof that the limit $\lim_{n \to \infty} E_n$ yields the ground state energy goes along with proving the convergence of the ground state projections corresponding to $H_n$, for every $n$, {\color{black}\secret{which}} is the main part of our proof. These projections, at each step $n$, are proved to exist and to be rank-one, for sufficiently small $g$ uniformly in $n$: Eq.~\eqref{int1} permits us to calculate the projection associated to $E_n$ using Riesz integrals, since it implies that $E_n$ is isolated. We actually define, for every $n \in \mathbb{N}$,
\begin{align} \label{gamman}
\Gamma_n := \Big\{ z \in \mathbb{C} \, \, \Big\vert \, \, \vert z - E_n \vert = \frac{1}{8} \rho_{n} \Big\} \, 
\end{align}
and 
\begin{equation}\label{pene}
P_n : =  \frac{-1}{2\pi i } \int_{\Gamma_n} \frac{dz}{H_n - z}. 
\end{equation}
It is not difficult to prove that $P_n$ is a rank-one projection, for
every $n$. It follows from the fact that, for sufficiently small $g$
({\color{black}\secret{independently of} uniformly in} $n$), $\| P_{n +1}
- P_n \| $ is {\color{black} strictly} smaller than $1$ and $P_0$ is
rank-one -- see Corollary \ref{rankone}. As {\color{black}\secret{we}}
mentioned above, this is proved without using that $G$ is
off-diagonal.

The most difficult part of the paper is to prove that the sequence of projections 
$ \{ P_n \}_{n \in \mathbb N} $ converges (the range of the limiting projection actually consists of ground state eigenvectors of the original Hamiltonian $H$). Of course, the projections $   P_n $, for $n \in \mathbb{N}$, act on different Hilbert spaces, but we identify them with projections acting on the full Hilbert space $\mathcal{H}$ by applying tensor products with the vacuum state projections on 
{\color{black} \begin{align}
\mathcal{F}_n^\infty := \mathcal{F}[L^{2}(B_{n})] \, , \label{eq.2.2} 
\end{align}
}i.e., we define 
\begin{equation}\label{extra}
P^\infty_n := P_n \otimes \Big ( | \Omega_{n}^\infty \rangle \langle \Omega_{n}^\infty | \Big ),
\end{equation}    
where $ \Omega_{n}^\infty $ is the vacuum in $ \mathcal{F}_n^\infty  $, notice that 
$\mathcal{H} = \mathcal{H}_n \otimes \mathcal{F}_n^\infty $.  

In Theorem \ref{MAIN} we prove that, for sufficiently small $g$ and $\gamma$,  
\begin{equation}\label{main1}
\|  P^{\infty}_{n + 1} - P^{\infty}_n    \| \leq  \Big ( \frac{1}{2} \Big )^{n}.
\end{equation}
Observe that the bound above is exponentially {\color{black}
 \secret{decreasing} small in $n$}. Actually, in the irrelevant case --
if the factor is $ 1/|k|^{1/2 - \mu} $, for some $\mu > 0$, instead of
$ 1/|k|^{1/2} $ -- a positive power of $\rho_n$ appears multiplying
the right side of Eq.~\eqref{main1} (see Remark \ref{reg}). Of course,
this term makes the calculations much simpler and direct (actually if
we {\color{black}\secret{would assume} had assumed} an infrared regular
interaction, Section \ref{regular} would basically contain the proof
of our main result -- see Remark \ref{reg}). {\color{black}\secret{Our
    situation is more difficult} The present situation is more
  complicated} and a {\color{black} more} subtle argument is
required. {\color{black}\secret{This is when the symmetry on $G$ (i.e.,
    it is off-diagonal) plays a role} At this point the symmetries $H$
  and $H_n$ possess, play a key role. Namely, with the help of a
  symmetry operator $\mathcal{S}$, see \eqref{mathcals}, which we
  prove to commute with $H_n$ ($n \in \mathbb{N}_0$), i.e., we
  identify a new conserved quantity of the model. The symmetry
  $\mathcal{S}$ is used to prove\secret{We prove}} that the fact that
$\sigma_1$ maps the ground state eigenspace corresponding to $H_{\rm
  at}$ to its orthogonal complement {\color{black}\secret{is preserved
    by} holds also true for every member of} the sequence $\{ H_{n}
\}_{n =0 }^{\infty} $ of operators, i.e., {\color{black}\secret{we prove
    that}}
\begin{equation}\label{simm}
P_n \sigma_1 P_n = 0,
\end{equation}
for every $n$, see Lemma \ref{key}. {\color{black}\secret{This is achieved with the help of a symmetry operator $\mathcal{S}$, see \eqref{mathcals}, which we prove to commute with $H_n$ ($n \in \mathbb{N}_0$), i.e., we identify a new conserved quantity of the model.}} 

The proof of Eq.~\eqref{main1} is technical and concentrated in Section \ref{keyestimates}, see Lemmas \ref{lemma.7}, \ref{Onesimo} and \ref{Onesimo2}. These Lemmas are collected in the proof of Theorem \ref{MAIN}, which is our principal demonstration.

Eq.~\eqref{main1} implies that the sequence of projections converges, provided that we choose $g > 0$ sufficiently small.  
Setting 
\begin{equation}
P_{\rm gs} : = \lim_{n \to \infty } P^{\infty}_n, 
\end{equation}
we observe that $P_{\rm gs}$ is a rank-one projection (being the limit
of rank-one projections) and, most importantly: Any non-zero vector in
the range of $ P_{\rm gs}$ is an eigenvector of $H$ corresponding to
the eigenvalue $E_{\rm gs}$. {\color{black} While this is our main
  result, we do not give a proof of the simplicity of the eigenvalue
  $E_{\rm gs}$ here. Note, however, that the semigroup $e^{-\beta H}$
  generated by $H$ is known to be positivity improving (in a suitable
  representation) \cite{HirokawaHiroshimaLorinczi2014}, and from this
  the uniqueness (non-degeneracy) of the ground state follows from a
  standard Perron-Frobenius argument, see, e.g.,
  \cite[Thm.~XIII.44]{ReedSimonIV1978}.}

\subsection{Prospective Generalizations}

In this paper we assume that the interaction between the atom and the
photon field is off-diagonal, and we restrict ourselves to a two-level
atom. {\color{black}\secret{The restriction to a two-level atom is not
    essential: A generalization to an $n-$level atom is
    straightforward, once we assume suitable symmetries generalizing
    the off-diagonal symmetry of the two-level case. We do not carry
    out this generalization here, because we believe it would not
    bring new insights and, on the contrary, it would complicate the
    presentation. The important restriction of our model is assuming
    the symmetry we referred above on the interaction. However, we
    believe it is unavoidable, because there are several negative
    results and strong requirements on putative ground states, see
    \cite{AraiHirokawaHiroshima1999}.} The generalization to an $N$-level atom is, however,
  not straightforward, because the mere existence of a symmetry
  $\mathcal{S}_N$ similar to the symmetry $\mathcal{S}_2 \equiv
  \mathcal{S}$ of the two-level atom implies severe and unphysical
  restrictions on the structure of the coupling function $G$. Indeed,
  the transcription of the proof of Lemma~\ref{key} would require the
  symmetry $\mathcal{S}_N$ to be invertible and commuting with $H$ and
  the $N \times N$ coupling function $G(k)$ to be similar (as a
  matrix) to $-G(k) = \mathcal{S}_N G(k) \mathcal{S}_N^{-1}$ when
  conjugated with $\mathcal{S}_N$. This is not surprising because of
  several known negative results and strong requirements on putative
  ground states, see \cite{AraiHirokawaHiroshima1999}.}

{\color{black} For a coupling function $G$ with a bipartite structure,
  these requirements are fulfilled. Bipartiteness means that the
  atomic energy levels form two disjoint sets, $A$ and $B$, say, and
  level transitions $A \to A$ and $B \to B$ are forbidden. (For the
  two-level atom considered here, $A = \{ 0 \}$ and $B = \{2\}$ and
  bipartiteness simply means that there are no self-interactions of
  the atomic orbitals.) There is no physical reason that would justify
  this assumption, in general. Nevertheless, the proof of binding
  presently given can be easily transcribed to the general bipartite
  situation}

{\color{black} As established here for two-level atoms or elsewhere for
  other models of NR~QED, \textit{binding} states that the ground
  state energy is an eigenvalue and the ground state vector is an
  element of the Hilbert space $\mathcal{H} = \mathcal{H}_\at \otimes
  \mathcal{F}$ which carries a Fock representation of the canonical
  commutation relation. We believe that, following an argument
  originating in work by Fr\"ohlich \cite{Froehlich1973} and Pizzo
  \cite{Pizzo2003,Pizzo2005} and further developed by Chen and
  Fr\"ohlich \cite{ChenFroehlich2007} and by Matte and one of us
  \cite{KoenenbergMatte2014}, it is possible to establish binding in a
  more general (and weaker) sense, and we now outline how this could
  be done on the example of the \textit{Generalized
    Spin-Boson-Hamiltonian} given by
\begin{equation}\label{eq1}
\widehat{H} = H_{ \rm at}+ H_{\mathrm{ph}}(\omega)+\Phi(\widehat{G}),
\end{equation}
which is an operator on 
\begin{equation}
\mathcal{H}=\CC^d\otimes \mathcal{F}.
\end{equation}
The atomic Hamiltonian, $H_{ \rm at}$, is a diagonal self-adjoint
$d\times d$-matrix whose lowest eigenvalue is simple and $\widehat{G}$
is of the form
\begin{align}
\widehat{G}(k) :=  g \frac{\Lambda(k)}{ 4 \pi \sqrt{\omega(k)}} f(k)  \, M
 \, ,  \ \forall k \in \mathbb{R}^3 \, , \label{eq.7a}
\end{align}
where $f$ is as before, and $M$ is a self-adjoint $d\times
d$-matrix. Similar to the method applied in this paper, we define an
infrared-regularized Hamiltonian $\widehat{H}_n$ on
$\mathcal{H}_n=\CC^d\otimes \mathcal{F}_n$ by replacing $\widehat{G}$
by $\widehat{G}_n:= {\bf 1}_ { \mathbb{R}^3 \backslash B_n } \,
\widehat{G}$ and $\omega$ by $\omega_n$ in
\eqref{eq1}. Proposition~\ref{gapgood} and its proof hold for
$\widehat{H}_n$ mutatis mutandis. In particular, if $g$ is
sufficiently small then there exists a unique normalized ground state
$\widehat{\phi}_n$ of $\widehat{H}_n$, for every $n \in \NN_0$.}

{\color{black} As opposed to the sequence $\{ \phi_n \}_{n = 0}^\infty$
  of ground states analyzed in this paper, the sequence 
  $\{ \widehat{\phi}_n \}_{n = 0}^\infty$ of ground states does not
  converge (strongly), but $ \widehat{\phi}_n \to 0$ weakly, as $n \to
  \infty$. Yet, as a state on $\bigcup_{m=1}^\infty \cB[\mathcal{H}_m] \otimes \mathds{1}_{\mathcal{F}_m^\infty}  \ni A$, 
\begin{equation}\label{mutatismutandis}
\hat{\omega}(A) 
\ = \ 
\lim_{n\to\infty} \langle \widetilde{\phi}_n, \: A \,\widetilde{\phi}_n\rangle.
\end{equation}
does exist, using $\widetilde{\phi}_n=\widehat{\phi}_n\otimes
\Omega_n^\infty$, cf.~\eqref{extra}. This limit state can be
represented as the GNS-vector in a non-Fock representation of the
CCR-algebra. The absence of binding in the strict sense is reflected
here in the fact that there is no vector $\widehat{\phi}_{\rm gs}$
(nor density matrix) in the original Hilbert space $\mathcal{H}$ such
that $\hat{\omega}(A) = \langle \widehat{\phi}_{\rm gs}, \: A \,
\widehat{\phi}_{\rm gs} \rangle$.

Yet, the nature of the limit in \eqref{mutatismutandis} can be made
more precise. Namely, the conjugation of $\widehat{H}_n$ by a unitary
operator $U_n$,
\begin{equation}\label{mutatismutandis-2}
K_n \ := \ U_n \, \widehat{H}_n \, U_n^* \otimes 
|\Omega_n^\infty \rangle\langle \Omega_n^\infty| , 
\quad 
\Psi_n \ := \ 
U_n \, \widehat{\phi}_n \: \otimes \: \Omega_n^\infty ,
\end{equation}
for each $n$, yields new sequences $\{ K_n \}_{n=0}^\infty$ and $\{
\Psi_n \}_{n=0}^\infty$ of Hamiltonian operators on $\mathcal{H}$ with
ground state energies $E_n$ and unique normalized ground state vectors
$\Psi_n$. The main point is that there exists a sequence $\{ U_n
\}_{n=0}^\infty$ of suitably chosen Bogolubov transformations (in
fact, even Weyl transformations with a fairly explicit description)
such that $K_n  \to K$ converges in strong resolvent sense to
a self-adjoint operator $K$ on $\mathcal{H}$ and $\Psi_n \to \Psi_{\rm
  gs} \in \mathcal{H}$, as $n \to \infty$. The sequence $\{ U_n
\}_{n=0}^\infty$ of Bogolubov transformations, however, does not
converge, and even though $K$ can be formally obtained from a shift
$a(k) \mapsto a(k) + h(k)$, for a suitable function $h:
\mathbb{R}\setminus\{0\} \to \mathbb{C}$, this shift is not unitarily
implementable, i.e., there is no unitary operator $U$ on $\mathcal{H}$
such that $K = U \widehat{H} U^*$. Nevertheless, one may argue that
$K$ is the new, renormalized Hamiltonian describing the physics
(generating the actual dynamics).  }

\subsection{Acknowledgements}

\noindent M. Ballesteros is a fellow of the Sistema Nacional de Investigadores (SNI). The research work of M. Ballesteros is partially supported by the project PAPIIT-DGAPA UNAM IN102215. \\
\noindent Part of this work was carried out while M. Könenberg was a Postdoctoral Fellow in the Department of Mathematics and Statistics at Memorial University of Newfoundland, Canada, where he was supported by an NSERC Discovery Grant and an NSERC Discovery Grant Accelerator.\\
\noindent The work of M. Könenberg was partially supported by the Deutsche Forschungsgemeinschaft (DFG) through the Research Training Group 
   1838:  Spectral Theory and Dynamics of Quantum Systems.

\section{The Sequence of Infrared-Regular Ground State Energies and Projections}\label{regular}

In this section we study spectral properties of $H_n$, for every $n \in \mathbb{N}_0$.  We start with a brief notation section (see Section \ref{not}, in which we also state some standard results). Then we estimate the distance between consecutive spectral points $E_{n+1}$ and $E_n$, see Proposition \ref{energies}. Right after, we prove that $E_n$ is isolated from the rest of the spectrum of $H_n$, for every $n \in \mathbb{N}_0$. This is achieved in Proposition \ref{gapgood}, which quantifies the gap above the ground state energy of $H_n$ and is a direct consequence of Lemma \ref{gapst}. The main technical tool in the proof of both, Proposition \ref{energies} and Proposition \ref{gapgood}, is Lemma \ref{tecnico}.  Of course, the gap estimates ensure the existence of ground state projections, see 
Eqs.~\eqref{pene} and \eqref{penetilde}. An additional effort permits us to estimate the norm difference of projections $P_n$ and $P_{n+1}$, where $P_n$ denotes the projection onto the ground state eigenspace of $H_n$. This is derived in Proposition \ref{Pro.reg.main}, which due to $ \| P_n - P_{n+1} \| < 1 $ implies that all projections $P_{n}$, $n \in \mathbb{N}_0$, are rank-one (see Corollary \ref{rankone}). 
In Section \ref{symmetry} we present a new conserved quantity in the spin-boson model with off-diagonal interaction. We prove that the fact that  $\sigma_1$ maps the ground state eigenspace corresponding to $H_{\rm at}$ to its orthogonal complement is preserved by the flow of operators 
$\{ H_{n} \}_{n  \in \mathbb{N}_0} $, i.e., we prove that $P_n \sigma_1 P_n = 0$, for all $n$, see Lemma \ref{key}. This is achieved with the help of an operator $\mathcal{S}$, see \eqref{mathcals}, which we prove to commute with $H_n$ ($n \in \mathbb{N}_0$) and hence identify a new conserved quantity of the model.

%
%
\subsection{Notation and Standard Results}\label{not}
For every normed vector space $V$, we denote by $\| \cdot  \|_{V}$ its norm. If $V$ is a Hilbert space, we denote by $ \langle \cdot | \cdot \rangle_{V} $ its inner product. If it is clear, however, from the context, we omit the subscripts $V$.

 We introduce some basic notation that we use for the construction of the sequence of eigenvalues $\{E_n \}_{n \in \mathbb{N}_0} $ and ground state projections $\{P_n \}_{n \in \mathbb{N}_0} $. 
Recalling that 
\begin{align}
B_n : = \{  k \in \mathbb{R}^{3}  | |k| < \rho_n \} \subset \mathbb{R}^{3}, \quad
\rho_n = \kappa \gamma^{n},
\end{align}
we introduce
\begin{align}\label{widetildes}
\widetilde \omega_n (k) := &  {\bf 1}_ { B_{n} \backslash B_{n+1} } \, \omega \, , & \widetilde  G_n (k) := &    {\bf 1}_ { B_{n }\backslash B_{n+1} }  \, G \,  
\end{align}
and the Fock spaces
\begin{align}
\widetilde{\mathcal{F}}_n := \mathcal{F} [L^{2}[B_{n} \backslash B_{n+1 } ] ],  \label{eq.2.3} 
\end{align}
with vacuum states
\begin{equation}\label{vac}
\Omega_{ L^{2}[B_{n} \backslash B_{n+1 } ] }  \equiv \widetilde \Omega_n.
\end{equation}
The projections onto the one-dimensional subspaces generated by the vectors 
$\Omega, \Omega_n $   and  $   \widetilde \Omega_n$ are denoted by
\begin{align}\label{pomega}
P_{\Omega },  \hspace{2cm}P_{\Omega_n },  \hspace{2cm}P_{\widetilde \Omega_n },  
\end{align}   
respectively.
We define
\begin{align}
\widetilde H_{n} :=  H_n \otimes \mathds{1}_{\widetilde { \mathcal{F}}_n} + \mathds{1}_{\mathcal{H}_n} \otimes H_{ph}(\widetilde \omega_n) \, ,  \label{eq.2.3.a}
\end{align}
as operators on (a suitable domain in) $\mathcal{H}_{n+1}$.
We observe that $ \inf \, \sigma( H_{n} )   =  \inf \, \sigma( \widetilde H_{n} )   $ and denote
\begin{align}
E_{n} := \inf \, \sigma( H_{n} )   =  \inf \, \sigma( \widetilde H_{n} )    \, .
\end{align}
The distance  (gap) from $E_n$ to the rest of the spectrum of $H_{n}$ (respectively 
$\widetilde{H}_n $) is given by
\begin{align}
\gap_{n} := \inf \big\{ \sigma( H_{n} ) \backslash \{ E_{n} \} \big \} - E_n 
\end{align}
and
\begin{align}
\widetilde{\gap}_{n} := \inf \big\{ \sigma( \widetilde{H}_{n} ) \backslash \{ E_{n} \} \big \} - E_n \, ,
\end{align}
respectively. The following basic estimate is frequently used in this paper (see \cite{BachFroehlichSigal1998a,BachFroehlichSigal1998b} for a proof):     
\begin{lemma}\label{basic.estimate.1}
Let {\color{black}\secret{$\nu > 0$} $\rho > 0$} be arbitrary. For all $F \in L^2(\mathbb{R}^3 ; \mathbb{C})$ with
$\omega^{-\frac{1}{2}}F \in L^2(\mathbb{R}^3 ; \mathbb{R})$,
{\color{black}
\begin{align}
\Big \Vert \Phi(F ) \Big (H_\rad(\mathds{1}_{supp(F)} \, \omega) + \rho \Big )^{-\frac{1}{2}}  \Big \Vert
 & \leq 2 \Big (
 \Vert \omega^{-1/2}  {F} \Vert_{L^2} + \rho^{-1/2} \Vert   {F} \Vert_{L^2} \Big )
 \: , \label{basic-estimate.3}
\end{align}
}where $\Phi(F) $ is defined as in \eqref{eq.7}. 
\end{lemma}
Since we assume $\Vert f \Vert_{L^\infty} \leq  1$, we immediately get by Lemma \ref{basic.estimate.1} that, for every $n \in \mathbb{N}_0$,
\begin{align}
\Big \Vert \Phi(\widetilde G_n) \Big (H_\rad(\widetilde \omega_n) + \rho_n 
\Big )^{-\frac{1}{2}}  \Big \Vert
  \leq  g \rho_{n}^{\frac{1}{2}}
 \: . \label{basic-estimate.3a}
\end{align}
%
%
%
%

\subsection{The Sequence of Ground State Eigenvalues and Projections}\label{inte}

In this section we estimate the distance between consecutive spectral infima $E_{n+1}$ and $E_n$, see Proposition \ref{energies}. Right after we prove that $E_n$ is isolated from the rest of the spectrum of $H_n$, for every $n \in \mathbb{N}_0$. This is achieved in Proposition \ref{gapgood}, which is a direct consequence of Lemma \ref{gapst}. 
 Of course, the gap estimates ensure the existence of ground state projections, see Eqs.~\eqref{pene} and \eqref{penetilde}. An additional effort permits us to estimate the norm-difference of consecutive projections $P_n$ and $P_{n+1}$ in Proposition \ref{Pro.reg.main}. In particular we prove that $\|P_{n} - P_{n+1} \| < 1$, which implies that all projections $P_{n}$, $n \in \mathbb{N}_0$, are rank-one (see Corollary \ref{rankone}).

\begin{lemma}\label{tecnico}
For every $n \in \mathbb{N}_0$, 
\begin{align}\label{mtat}
 H_{n+1} + \rho_n   \geq H_n +  (1 -   g  )
   \big ( H_\rad(\widetilde \omega_n ) + \rho_n \big  )  
\end{align}
holds true in the sense of quadratic forms.

\end{lemma}
\noindent{\it Proof: }
 Let $\psi \in \mathcal{H}_{n +1}$ be a normalized vector in the domain of $H_{n+1}$. We calculate
\begin{align}\label{mta3}
 \langle  \psi | H_{n+1} \psi \rangle =  \langle  \psi | H_{n} \psi \rangle +  \langle  \psi | \Phi(\widetilde G_n) \psi \rangle +  \langle  \psi | H_\rad(\widetilde \omega_n ) \psi \rangle .
\end{align}
Next, we  set 
\begin{equation}\label{mta3.1} A  :=    1 +   \big ( H_\rad(\widetilde \omega_n ) + \rho_n \big  )^{-1/2} \Phi(\widetilde G_n)
     \big ( H_\rad(\widetilde \omega_n ) + \rho_n \big  )^{-1/2}   
\end{equation}
and notice that      
\begin{align}\label{mta4}
 \langle  \psi |  \Phi(\widetilde G_n) \psi \rangle + &  \big \langle  \psi \big |  \big ( H_\rad(\widetilde \omega_n ) + \rho_n \big  ) \psi 
  \big \rangle \\ = \,  \notag &
   \Big \langle \big ( H_\rad(\widetilde \omega_n ) + \rho_n \big  )^{1/2}  \psi \Big |  A   \big ( H_\rad(\widetilde \omega_n ) + \rho_n \big  )^{1/2} \psi 
  \Big \rangle.
\end{align}
As (see Eq.~\eqref{basic-estimate.3a})
\begin{align}
\Big \|  \big ( H_\rad(\widetilde \omega_n ) + \rho_n \big  )^{-1/2} \Phi(\widetilde G_n)
     \big ( H_\rad(\widetilde \omega_n ) & + \rho_n \big  )^{-1/2}     \Big \| 
   \notag \\ \notag  &\leq  \rho_n^{- \frac{1}{2}}
 \Vert
\Phi(\widetilde G_n)
     \big ( H_\rad(\widetilde \omega_n ) + \rho_n \big  )^{-1/2}  \Vert  
    \notag  \\ & \leq   g, 
\end{align}
we obtain that 
\begin{equation}\label{mta5}
A \geq 1 -   g ,
\end{equation}
and, therefore, using \eqref{mta3.1}  we get 
\begin{align}\label{mta6}
 \langle  \psi |  \Phi(\widetilde G_n) \psi \rangle + &  \big \langle  \psi \big |  \big ( H_\rad(\widetilde \omega_n ) + \rho_n \big  ) \psi 
  \big \rangle \\ \notag & \geq (1 -   g  )
  \langle  \psi \big |  \big ( H_\rad(\widetilde \omega_n ) + \rho_n \big  ) \psi 
  \big \rangle . 
\end{align}
Eqs.~\eqref{mta3} and \eqref{mta6} imply Eq.~\eqref{mtat}.

\qed

The same  argument as in the proof of Lemma \ref{tecnico} yields the following quadratic form estimate.   
\begin{lemma}\label{nuevos}
For every $n \in \mathbb{N}_0$, we have that 
\begin{align}\label{mtatt}
 H + \rho_n   \geq H_n +  (1 -    g  )
   \big ( H_\rad(\widetilde \omega_n ) + \rho_n \big  ).  
\end{align} 

\end{lemma}

\begin{proposition}[Energy Differences] \label{energies}
Suppose  $  g  < 1  $. 
For every $n \in \mathbb{N}_0$, we have that
\begin{equation}\label{mta0}
|E_{n+1} - E_n| \leq  g \rho_n .
\end{equation}
\end{proposition}
\noindent{\it Proof:}
First notice that, for every normalized vector $\phi \in \mathcal{H}_n$ in the domain of $H_n$,  $ \psi = \phi \otimes \widetilde \Omega_n \in \mathcal{H}_{n +1}$ and 
\begin{align}\label{mta1}
E_{n+1 } \leq  \langle  \psi | H_{n+1} \psi \rangle_{\mathcal{H}_{n+1}}  =  
 \langle  \phi | H_n \phi \rangle_{\mathcal{H}_{n}}. 
\end{align}
Taking the infimum over such $ \phi's  $ we get
\begin{equation}\label{mta2}
E_{n+1} \leq 	E_n. 
\end{equation}
Now we take a normalized vector $\psi $ in the domain of $H_{n+1}$.  We notice that (we recall that $ g < 1$) 
\begin{align}\label{mta6primas}
 (1 -    g  )
   \big ( H_\rad(\widetilde \omega_n ) + \rho_n \big  )   \geq  (1 -    g  ) \rho_n  
\end{align}
and use Lemma \ref{tecnico} to obtain 
\begin{align}\label{mta7}
 \langle  \psi | H_{n+1} \psi \rangle_{\mathcal{H}_{n+1}}  \geq 
  \langle  \psi | H_{n} \psi \rangle_{\mathcal{H}_{n+1}}  -    g \rho_n \geq E_n -   g \rho_n,
\end{align} 
from which we get 
\begin{align}\label{mta8}
E_{n+1} \geq E_n -   g \rho_n. 
\end{align} 
Eqs.~\eqref{mta2} and \eqref{mta8} imply \eqref{mta0}.

\qed

\begin{lemma} \label{gapst}
Suppose  $    g  < \frac{1}{2} \gamma $. 
For every $n \in \mathbb{N}_0$, we have that 
\begin{equation}\label{ok0}
\gap_{n+1} \geq \min \big (\gap_n ,  (1 -    g  )  \rho_{n + 1} \big ) -   g \rho_n .
\end{equation}
\end{lemma}
\noindent{\it Proof: }
We use the min-max principle to estimate $\gap_{n+1}$ as
\begin{align}\label{ok1}
\gap_{n + 1} = \sup_{\psi \in \mathcal{H}_{n+1} \backslash \{0 \} }  \inf_{\phi \perp \psi, \Vert \phi \Vert =1} \langle \phi \big | (H_{n+1} - E_{n+1}) \phi \rangle , 
\end{align}
where $\phi $ is additionally assumed to lie in the domain of $H_{n+1}$.  We take $\phi $ as in Eq.~\eqref{ok1} and  
 utilize Lemma \ref{tecnico} to obtain : 
\begin{align}\label{ok2}
\langle \phi | (H_{n+1} - E_{n+1}) \phi \rangle \geq & \, \langle \phi | (H_{n} - E_{n}) \phi \rangle + E_n - E_{n +1} - \rho_n   \\  \notag
&  + \, 
   \Big \langle  \phi \Big |   (1 -    g  )
   \big ( H_\rad(\widetilde \omega_n ) + \rho_n \big  )  \phi 
  \Big \rangle
  \\ \notag = &   \Big \langle  \phi \Big | \Big [ H_n - E_n +  (1 -   g  )
    H_\rad(\widetilde \omega_n )   \Big ]  \phi 
  \Big \rangle \\ \notag &  + E_n - E_{n +1} -   g  \rho_n. 
\end{align}
We temporarily set
\begin{equation}\label{temp}
Q := H_n - E_n +  (1 -    g  )
  H_\rad(\widetilde \omega_n )
\end{equation}
and observe that $\inf \sigma(Q) =0$. Denoting by $\gap(Q)$ the distance between $0$ and the rest of the spectrum of $Q$, we arrive at 
\begin{equation}\label{cas}
\gap_{n +1} \geq \gap(Q)  + E_n - E_{n +1} -   g  \rho_n,
\end{equation}
where we use \eqref{ok2}, \eqref{temp}, and the min-max principle applied to $H_{n+1}$ and $Q$. Using the fact that $H_{n}$ and $  H_\rad(\widetilde \omega_n )  $ act on different factors in the tensor product decomposition $\mathcal{H}_{n+1}
=\mathcal{H}_{n} \otimes \widetilde{\mathcal{F}}_{n} $, we readily get  
\begin{equation}\label{gapq}
\sigma(Q) \subset \{ 0 \} \cup \Big [ \min \big (\gap_n ,  (1 -    g  )  \rho_{n + 1} \big ), \infty \Big ) 
\end{equation}
(recall that $ g <1$) and, therefore, 
\begin{equation}\label{gapq2}
\gap(Q) \geq \min \big (\gap_n ,  (1 -    g  )  \rho_{n + 1} \big ). 
\end{equation}
Eqs.~\eqref{mta2}, \eqref{cas} and \eqref{gapq2} imply Eq.~\eqref{ok0}, here we use that $  g < \frac{1}{2}\gamma$.

\qed

\begin{remark}\label{gaptilde}
We notice that the spectral theorem directly implies that,  
for every $n \in \mathbb{N}_0$: 
\begin{equation}\label{ok4}
\widetilde{\gap}_{n} = \min \big ( \gap_{n },  \rho_{n + 1} \big ).  
\end{equation}
\end{remark}
As $G_0 = 0$, the spectrum of $H_{0}$ can be  calculated explicitly, with the result that 
\begin{align}
\sigma (H_0) = \{ 0  \} \cup [ \rho_0, \infty ),  
\end{align} \label{step0}
and, therefore, 
\begin{align}\label{gap0}
\gap_0 = \rho_0 = \kappa. 
\end{align}
We simplify Eq.~\eqref{ok0}  by assuming some hypothesis on $g$ and $\gamma$. Taking, for example, $  g < \frac{1}{4} \gamma  $ and $\gamma < \frac{1}{2}$, we inductively obtain, from \eqref{ok0}, that
$
\gap_{n } \geq \frac{1}{2} \rho_n$, for all $ n \in \mathbb{N}_0$. It also follows, from Remark \ref{gaptilde}, that $ \widetilde{\gap}_{n} = \rho_{n+1}  $.
Then, we arrive at the following proposition. 
\begin{proposition}[Gaps] \label{gapgood}
Suppose that $  g < \frac{1}{4} \gamma  $ and $\gamma < \frac{1}{2}$. Then 
\begin{equation}\label{ssss}
\gap_{n } \geq \frac{1}{2} \rho_n, \hspace{1cm}   \widetilde{\gap}_{n} = \rho_{n+1} , 
\end{equation}
for all $n \in \mathbb{N}_0$.
\end{proposition}
Proposition \ref{gapgood} permits us to define $P_n$  as in \eqref{pene}. It also allows defining 
\begin{equation}\label{penetilde}
\widetilde P_n : =  \frac{-1}{2\pi i } \int_{\Gamma_{n+1}} \frac{dz}{ \widetilde H_n - z},
\end{equation}
where the contour $\Gamma_n$ is defined in Eq.~\eqref{gamman}, provided $g > 0$ obeys $   g < \frac{1}{64}  \gamma   $, because in this case $E_n$ is the only spectral point of $ \widetilde H_n $ in the interior of $\Gamma_{n +1}$ (see Propositions \ref{energies} and \ref{gapgood}). 
Notice that
\begin{equation}\label{impo}
\widetilde P_n = P_n \otimes \widetilde P_{\Omega_n} . 
\end{equation}
In the next lemma we estimate the norm-difference of
$ P_{n+1} $  and $\widetilde P_n $. 
\begin{proposition} \label{Pro.reg.main} Suppose that $  g < \frac{1}{64} \gamma  $ and $\gamma < \frac{1}{2}$. Then
\begin{equation}\label{projections}
 \| P_{n+1} - \widetilde P_n \| \leq   \frac{16}{\gamma}  g \leq \frac{1}{4}.
\end{equation}
\end{proposition}
\noindent{\it Proof:}
The second inequality in \eqref{projections} is obvious. 
We use the second resolvent identity and \eqref{pene} and \eqref{penetilde} to get
\begin{align}\label{pro1}
P_{n+1} - \widetilde P_n = \frac{1}{ 2 \pi i} \int_{\Gamma_{n+1}} 
  \frac{1}{H_{n+1} -z}   \Phi(\widetilde G_n)  \frac{1}{\widetilde H_{n} -z} dz.   
\end{align} 
Next we estimate 
\begin{align}\label{we}
\Big \|     \Phi(\widetilde G_n)  \frac{1}{\widetilde H_{n} -z}  \Big \| & \leq
\big \|     \Phi(\widetilde G_n)   \big ( H_\rad  (\widetilde \omega_n) + \rho_n  \big )^{-1/2}  \big \|
 \\ \notag & \hspace{3cm} \cdot \Big \| \big ( H_\rad(\widetilde \omega_n) + \rho_n  \big )^{1/2}   \frac{1}{\widetilde H_{n} -z}  \Big \| 
 \\ \notag & \leq   g \rho_{n }^{1/2}\Big \| \big ( H_\rad(\widetilde \omega_n) + \rho_n  \big )^{1/2}   \frac{1}{\widetilde H_{n} -z}  \Big \|, 
\end{align}
where we use \eqref{basic-estimate.3a}. Functional calculus of self-adjoint operators allows us to compute  
\begin{equation}\label{estap}
\Big \| \big ( H_\rad(\widetilde \omega_n) + \rho_n  \big )^{1/2}   \frac{1}{\widetilde H_{n} -z}  \Big \| = \sup_{s \in \sigma(H_n)}\sup_{r \in \{0\} \cup [ \rho_{n+1}, \infty ) } \Big | (r + \rho_n)^{1/2} \frac{1}{ s + r - z} \Big |. 
\end{equation}
The definition of $\Gamma_{n+1}$,
\begin{align} \label{gammanprima}
\Gamma_{n+1} := \Big\{ z \in \mathbb{C} \, \, \Big\vert \, \, \vert z - E_{n+1} \vert = \frac{1}{8} \rho_{n +1} \Big\} \,  
\end{align}
(see \eqref{gamman}),
Propositions \ref{energies} and \ref{gapgood}, and  $  g \rho_n \leq \frac{\gamma}{64}\rho_n =  \frac{1}{64}  \rho_{n+1}$, then lead us to
\begin{align}\label{non.nece}
|s - z| & \geq \frac{1}{16} \rho_{n+1}, \hspace{4cm} \forall s \in \sigma(H_n), 
\\
\notag | s+r - z| &  \geq \frac{r}{2} \geq \frac{1}{2} \Big (  \frac{r}{2} + \frac{\rho_{n+1}}{2}  \Big ), \hspace{2cm}  \forall s \in \sigma(H_n), r \in [\rho_{n+1}, \infty). 
\end{align}
Eqs.~\eqref{estap} and \eqref{non.nece} imply that
\begin{equation}\label{estap1}
\Big \| \big ( H_\rad(\widetilde \omega_n) + \rho_n  \big )^{1/2}   \frac{1}{\widetilde H_{n} -z}  \Big \| \leq \frac{16}{\gamma \rho_n^{1/2}}, 
\end{equation}
while \eqref{we} and \eqref{estap1} imply that 
\begin{equation}\label{casi}
\Big \|    \frac{1}{H_{n+1} -z}   \Phi(\widetilde G_n)  \frac{1}{\widetilde H_{n} -z}  \Big \| \leq  \frac{128  g}{\gamma \rho_{n+1}}.  
\end{equation}
We prove \eqref{projections} using \eqref{pro1}, \eqref{casi}, and the definition of $\Gamma_{n+1}$.

\qed

\begin{corollary}\label{rankone}
Suppose that $  g < \frac{1}{64} \gamma  $ and $\gamma < \frac{1}{2}$. Then $P_n$ is a rank-one orthogonal projection, for every $ n \in \mathbb{N}_0 $. 
\end{corollary}
\noindent{\it Proof: }
As $G_0 = 0$, it is straightforward to verify that $P_0$  and hence $\widetilde P_0$ are of rank-one. Proposition \ref{Pro.reg.main} implies that 
$\| P_1 - \widetilde P_0 \| < 1 $ and, therefore, $P_1$ is of rank-one, too. We proceed inductively to conclude that $P_n$ is of rank one for every $n \in \mathbb{N}_0$. The self-adjointness of $H_n$ ensures the self-adjointness of $P_n$.       
\qed

\begin{remark}\label{reg}
In case that $ G $ was infrared regular, i.e.,  if  $ \| |k|^{1/2 - \mu} G(k)\|$ was bounded at $k=0$, for some $\mu > 0$, then in Eq.~\eqref{we} and, consequently, in Eq.~\eqref{projections} we would gain a positive power of $\rho_n$. This would immediately imply the convergence of the projections $\{ P_n^{\infty} \}_{n \in \mathbb{N}_0}$, see Eq.~\eqref{extra}, which in turn implies the existence of the ground state. In other words: If we had an infrared regular interaction, this section would basically contain the proof of the existence of a ground state, our main result . Not assuming infrared regularity complicates the matter significantly. The next section addresses this complication.      
\end{remark}

\subsection{Invariant Subspaces Due to the Symmetry}\label{symmetry}

In this section we present a new conserved quantity in the spin-boson model with off-diagonal interaction. We prove that the fact that  $\sigma_1$ maps the ground state eigenspace corresponding to $H_{\rm at}$ to its orthogonal complement (it is off-diagonal) is preserved by the flow of operators 
$\{ H_{n} \}_{n  \in \mathbb{N}_0} $, i.e., we prove that $P_n \sigma_1 P_n = 0$, for all $n$, see Lemma \ref{key}. This is achieved with the help of a symmetry operator $\mathcal{S}$, see \eqref{mathcals} which we prove to commute with $H_n$ ($n \in \mathbb{N}_0$), i.e., it gives a new conserved quantity of the model.

The photon number operator defined on $ H_n $ is also denoted by $ \mathcal{N}_\rad = H_\rad(  {\bf 1}_ { \mathbb{R}^3 \backslash B_n })  $, see \eqref{numberfull}. Of course, $ \mathcal{N}_\rad $ depends on $n$, but we omit this in our notation. We use the standard representation for the Pauli matrices $\sigma_1, \sigma_1, \sigma_3 $ and denote by
$\mathcal{S}$ the following operator on $\mathcal{H}_n = \mathcal{H}_{\rm at} \otimes \mathcal{F}_n$ :  
\begin{align}\label{mathcals}
\cS : = \sigma_3  (-1 )^{\mathcal{N}_\rad}. 
\end{align}
\begin{lemma} \label{key}
For every $n\in \mathbb{N}_0$, we have that 
\begin{align}
P_n \sigma_1 P_n = 0 \, \, , \, \, \,  \widetilde{P}_n \sigma_1 \widetilde{P}_n =0 \, . \label{proj.tr.1}
\end{align}
\end{lemma}
\noindent{\it Proof:} The second equality in \eqref{proj.tr.1} follows from Eq.~\eqref{impo} and the first equality. We prove that $P_n \sigma_1 P_n =0$. 
A direct computation shows that   $ [ \sigma_3 , H_{\rm at}] = 0 $ (here $[\cdot, \cdot]$ denotes the commutator) . Furthermore, using the pull-through formulae,  
\begin{align} \label{pull}
a(k) \mathcal{N}_{ph} = (\mathcal{N}_{ph} +1) a(k) \ \ \ \text{and} \ \ \ a^{*}(k) ( \mathcal{N}_{ph} +1) =  \mathcal{N}_{ph} a^{*}(k) \, ,
\end{align}
we conclude that
\begin{align}\label{comm}
[\cS, H_n ] = 0. 
\end{align}
As $\mathcal{S}^2 = \mathds{1}$, we observe that  
\begin{equation}\label{comm1}
{\rm Tr}( \mathcal{S} P_n \sigma_1 P_n \mathcal{S} )  = {\rm Tr}(  P_n \sigma_1 P_n  ).
\end{equation}
Eq.~\eqref{comm} implies that $\cS$ commutes with $P_n $ and it is straightforward to verify  that it anti-commutes with $\sigma_1$. Hence we have  
\begin{equation}\label{comm2}
{\rm Tr}( \mathcal{S} P_n \sigma_1 P_n \mathcal{S} )  = -{\rm Tr}( \mathcal{S}^2 P_n \sigma_1 P_n  ) =  -{\rm Tr}(  P_n \sigma_1 P_n  ).
\end{equation}
Then we obtain from \eqref{comm1} and \eqref{comm2} that $  {\rm Tr}(  P_n \sigma_1 P_n  ) = 0$. Since $P_n$ is a rank-one projection, we conclude that $P_n \sigma_1 P_n =0$.  

\qed

\subsection{Further Estimates}\label{fur}
In this section we derive some estimates that are consequences of the computations in the present section, but will be used in our main section, Section \ref{Convergence}.   

\begin{lemma}\label{fur1}
Suppose that  $  g < \frac{1}{64} \gamma  $ and $\gamma < \frac{1}{2}$.
Take $z \in \Gamma_{n+1}$ (see Eq.~\eqref{gamman}).  The following norm bounds hold true,
\begin{align}
\Big \|\Phi(\widetilde G_n) \frac{1}{\widetilde H_n - z}   \Big \| \leq   \frac{16}{\gamma}  g \, , \label{basic-estimate.imp.1} \\ \notag \\ 
\Big \|\Phi(\widetilde G_n) \frac{1}{ H_{n+1} - z}   \Big \| \leq \frac{32}{\gamma}    g \, .
\end{align}
\end{lemma}
\noindent{\it Proof:}
We use that, see \eqref{basic-estimate.3a},
\begin{equation}
\Big \| \Phi(\widetilde G_n)  \Big (H_\rad 
\big (  \widetilde \omega_n) + \rho_n \Big ) ^{-1/2}   \Big \| \leq    g \rho_n^{1/2} \, ,
\end{equation}
and Eq.~\eqref{estap1} to prove the first inequality in \eqref{basic-estimate.imp.1}. To prove the second inequality we use a Neumann series,
\begin{align}\label{yano}
\Phi(\widetilde G_n) \frac{1}{ H_{n+1} - z}  = \Phi(\widetilde G_n) \frac{1}{ \widetilde H_{n} - z} \sum_{n=0}^{\infty} \Big ( - \Phi(\widetilde G_n) \frac{1}{ \widetilde H_{n} - z}  \Big )^n,  
\end{align}  
 the first inequality in \eqref{basic-estimate.imp.1}, and the fact that $   \frac{16}{\gamma} g \leq \frac{1}{2} $.    
\qed
\begin{corollary}\label{fur1c}
Suppose that  $  g < \frac{1}{64} \gamma  $ and $\gamma < \frac{1}{2}$.
Take $z $ in the interior of $ \Gamma_{n+1}$ (see Eq.~\eqref{gamman}). Then the following estimates hold true
\begin{align}
\Big \|\Phi(\widetilde G_n) \frac{1}{\widetilde H_n - z} (1 - \widetilde P_n)   \Big \| \leq \frac{16}{\gamma}    g , \\ \notag \\ 
\Big \|\Phi(\widetilde G_n) \frac{1}{ H_{n+1} - z}  (1 - P_{n+1}) \Big \| \leq \frac{32}{\gamma}    g. 
\end{align}
\end{corollary}
\noindent{\it  Proof:}
To prove the first inequality notice that  $  \frac{1}{\widetilde H_n - z} (1 - \widetilde P_n) $ is analytic in the interior of $\Gamma_{n+1}$. Then the claim follows from the maximum modulus principle. The second inequality is proved in the same way.
\qed

\section{Convergence of the Sequence of Ground State Projections}\label{Convergence}

In this section we prove our main result: We demonstrate that the sequence of ground state projections $(P_n^\infty)_{n \in \mathbb{N}_0}$ [see \eqref{extra}], converges and that the limit of it is a rank-one projection whose range consists of  ground state eigenvectors corresponding to $E_{\rm gs}$.  The convergence of the ground state projections is the content of Theorem \ref{MAIN}. The proof that the limit of this sequence corresponds to the ground state projection of $H$ is derived in Theorem \ref{GS-Projection}. The convergence of the sequence of the projections $P_n$ rests on the fact that $P_n \sigma_1 P_n = 0$, established in Lemma \ref{key}. On the technical level, this property is used in Eqs.~\eqref{sisepone} and \eqref{projdiff.5} below. The special difficulties of our proof come from the fact that the coupling function $G(k)$ behaves as $|k|^{-1/2}$, for small $k$, as explained before. In fact, for a more regular coupling function the results in Section \ref{regular} suffice (basically) to prove existence of the ground state. We start this section with introducing new notation in Section  \ref{NotationMain}.  Most of technical tools we need to prove our main results are collected in Section \ref{keyestimates}. A long line of arguments is split onto three lemmas:  Lemma \ref{lemma.7}, Lemma \ref{Onesimo}, and Lemma \ref{Onesimo2}. The idea is to bound the norm $ \|  P_{n+1} - \widetilde P_n \| $ in terms of the quantity $  \| R_n \sigma_1 P_n \| $, that can be recursively estimated in terms of $g$ and $ \gamma $. Lemmas \ref{lemma.7}, \ref{Onesimo} and \ref{Onesimo2} establish the main bounds we need to reach the recursive relation we are looking for. We put together all results of Section \ref{keyestimates} in the proof of Theorem \ref{MAIN}  in Section \ref{Main-Section}. Theorem \ref{MAIN} gives the convergence of the sequence of ground state projections $ \{  P^\infty_n \}_{n \in \mathbb{N}_0}$. The  limit of this sequence is denoted by  $ P_{\rm gs} $, its range consists of all ground state vectors corresponding to the fully interacting operator $H$, as it is demonstrated in Theorem \ref{GS-Projection}.

\subsection{Notation} \label{NotationMain}

For every projection $P$ we denote by $ P^\perp  : = 1 - P  $, the complement of P.   
We define
\begin{align}\label{n1}
R_{n}(z):= & (H_n - z )^{-1} \, , \\ 
\widetilde R_{n}(z):= & (\widetilde H_{n} -z )^{-1}  \, , \\ 
R_{n}(z)^\perp:= & R_n(z) P_n^\perp \, , \\ 
\widetilde R_{n}(z)^\perp := & \widetilde R_n(z) \widetilde P_{n}^\perp , 
\end{align}
whenever $z$ is not in the spectrum of the corresponding operator. If we project out  the eigenspace corresponding to $E_n$, we can take $z = E_n$ and set
\begin{align}\label{n2}
R_{n}(E_n)^\perp \equiv R_{n}^\perp \equiv & R_n P_n^\perp  \, , \\ \notag
\widetilde R_{n}(E_n)^\perp \equiv \widetilde R_{n}^\perp \equiv & \widetilde R_n   \widetilde P_n^\perp . 
\end{align}  
We finally introduce the function $\eta_n : \mathbb{R}^{3} \to \mathbb{C}$ by 
\begin{equation}\label{etan}
\eta_n(k) : =  : g  \:  {\bf 1}_ { B_{n }\backslash B_{n+1} } \frac{\Lambda(k)}{ 4 \pi \sqrt{\omega(k)}} f(k)  
 \, , 
\end{equation}
and note that (see \eqref{widetildes}) 
$$\widetilde G_n = \eta_n \sigma_1 .$$ 
We define the field operator $\Phi(\eta_n) : = a^*(\eta_n) + a(\eta_n)$, as in \eqref{eq.7}.

\subsection{Key Estimates}\label{keyestimates}
 Most of the technical tools we need to prove our main results are collected in this section.  The idea is to bound the norm $ \|  P_{n+1} - \widetilde P_n \| $ in terms of the quantity $ \| R_n^{\perp} \sigma_1 P_n\| $ that can be recursively estimated in terms of $g$ and $ \gamma $.  Lemmas \ref{lemma.7}, \ref{Onesimo} and \ref{Onesimo2}  establish the main bounds we need to reach the recursive relation we are looking for.

\begin{lemma}\label{lemma.7}
Suppose that $ g < \frac{1}{64} \gamma  $ and $\gamma < \frac{1}{2}$. It follows, for every $n \in \mathbb{N}_0$, that
\begin{align} \label{lemma71}
\Vert P_{n+1} - \widetilde P_{n} \Vert \leq 4 \Vert  P_{n+1}^\perp \widetilde P_{n} \Vert . 
\end{align}
\end{lemma}
\noindent{\it Proof:}
A direct computation shows that 
\begin{align}\label{m1}
P_{n+1} - \widetilde P_{n} = &  (P_{n+1} - \widetilde P_{n})^2 ({ \widetilde P_{n}^{\perp}} - \widetilde P_{n})  \\ \notag
& +  \widetilde P_{n} ( \widetilde P_{n}^\perp - P_{n+1}^\perp ) \widetilde P_{n}^\perp  +  \widetilde P_{n}^\perp  ( \widetilde P_{n}^\perp - P_{n+1}^\perp ) \widetilde P_{n} \, .
\end{align}
In fact, to prove \eqref{m1} we expand the right hand side of \eqref{m1} and use that 
\begin{align}
\widetilde P_n P_{n+1}  \widetilde P_{n}^{\perp} + \widetilde P_n P_{n+1}^{\perp}  \widetilde P_{n}^{\perp} = \widetilde P_n  \widetilde P_{n}^{\perp} = 0,
\end{align} 
 then we utilize the following identities 
\begin{align}
P_{n + 1}  \widetilde P_n^\perp  + \widetilde P_{n} P_{n+1} \widetilde P_n & - 
\widetilde P_n^\perp P_{n+1}^\perp \widetilde P_n \\ \notag & =  P_{n+1} (1 - \widetilde P_n )
+ \Big ( \widetilde P_n P_{n+1} - (1 - \widetilde P_n)(1 - P_{n+1})   \Big ) \widetilde P_n \\ \notag & = P_{n+1} + (-1 + \widetilde P_n )\widetilde P_n = P_{n+1}. 
\end{align}
Using Proposition \ref{Pro.reg.main} and \eqref{m1} we obtain 
\begin{align}
\Vert P_{n+1} - \widetilde P_{n} \Vert \leq & 2\Vert  ( P_{n+1} - \widetilde P_{n})^2 \Vert + 2 \Vert \widetilde P_{n}^\perp   ( \widetilde P_{n}^\perp -
  P_{n+1}^\perp ) \widetilde P_{n} \Vert \\
\leq &  \frac{1}{2} \Vert   P_{n+1} - \widetilde P_{n} \Vert   + 2 \Vert   P_{n+1}^\perp \widetilde P_{n} \Vert \, . \notag
\end{align}
Solving this for $   \frac{1}{2}\Vert   P_{n+1} - \widetilde P_{n} \Vert  $,
 we get \eqref{lemma71}.  

\qed

\begin{lemma}\label{Onesimo}
Suppose that $  g < \frac{1}{64}\gamma   $ and $\gamma < \frac{1}{2}$. It follows, for every $n \in \mathbb{N}_0$, that
\begin{equation}\label{onem}
 \| P_{n+1}^\perp \widetilde P_{n}  \| \leq 2 g \rho_n \Big \Vert R_{n+1}^\perp  \Big ( P_{n}^\perp   \otimes P_{\widetilde \Omega_n}^\perp \Big )  \sigma_1  \Big ( P_{n}  \otimes P_{\widetilde \Omega_n}^\perp \Big ) \Big \Vert . 
\end{equation}

\end{lemma}

\noindent{\it Proof:}
Multiplying and dividing by $  H_{n+1} - E_{n+1} $ and using that $(\widetilde H_n - E_n)\widetilde P_n = 0$, 
we get
\begin{align} \label{pl1}
P_{n+1}^\perp \widetilde P_{n} = \, & R_{n+1}^\perp [ H_{n+1} - E_{n+1} ] \widetilde P_{n} \\ \notag
= \, & R_{n+1}^\perp [  \Phi (\widetilde G_n ) + E_{n} - E_{n+1} ] \widetilde P_{n}\\ \notag
= \,  & R_{n+1}^\perp  \Phi (\widetilde G_n )  \widetilde P_{n} +  [E_{n} - E_{n+1} ] R_{n+1} P_{n+1}^\perp \widetilde P_{n} \, .
\end{align}
Propositions \ref{energies} and \ref{gapgood} and the spectral theorem yield
\begin{equation}\label{pl2}
\Big \| \Big ( E_{n} - E_{n+1} \Big ) R_{n+1} P_{n+1}^\perp  P_{n+1}^\perp \widetilde P_{n} \Big \| \leq
\frac{1}{16} \|   P_{n+1}^\perp \widetilde P_{n} \| . 
\end{equation}
Eqs. ~\eqref{pl1} and \eqref{pl2} imply that 
\begin{align} \label{otravez}
\Vert P_{n+1}^\perp \widetilde P_{n} \Vert \leq & 2 \Vert R_{n+1}^\perp  \Phi (\widetilde G_n )  \widetilde P_{n} \Vert \, .
\end{align}
The key symmetry property of our model implies that (see Lemma \ref{key})
 $  \Phi (\widetilde G_n ) P_n  $  $  =  $  $ P_n^\perp \Phi (\widetilde G_n ) P_n $, which in turn, together with the fact that $   \Phi (\widetilde G_n ) [\mathds{1}_{\mathcal{H}_n} \otimes P_{\widetilde \Omega_n} ] = [\mathds{1}_{\mathcal{H}_n} \otimes P_{\widetilde \Omega_n}^\perp ]  \Phi (\widetilde G_n ) [ \mathds{1}_{\mathcal{H}_n} \otimes P_{\widetilde \Omega_n} ]  $, gives (see also Eqs.~\eqref{impo} and  \eqref{eq.7prima}, \eqref{etan}) 
\begin{equation}\label{sisepone}
 \Phi (\widetilde G_n ) \widetilde  P_n = \Big ( P_{n}^\perp   \otimes P_{\widetilde \Omega_{n}}^\perp  \Big ) \,  
 \Phi (\widetilde G_n )  \widetilde  P_n
  = \Big ( P_{n}^\perp   \otimes P_{\widetilde \Omega_{n}}^\perp \Big ) \sigma_1 
  \Big ( P_{n}   \otimes P_{\widetilde \Omega_{n}}^\perp \Big ) \,  \Phi(\eta_n) \widetilde P_{n} .    
\end{equation}
Finally, we use that $ \eta_n $ is supported in $B_n$ and thus $ \Vert \Phi (\eta_n )   \widetilde P_n \| =  \|  P_n\| \cdot \|  a^*( \eta_n ) P_{\widetilde \Omega_n} \|  $   $ \leq g  \rho_n $-- see Eq.~\eqref{eq.4} -- 
and Eqs.~\eqref{otravez} and \eqref{sisepone} to arrive at Eq.~\eqref{onem}. 
\qed

\begin{lemma}\label{Onesimo2}
Suppose that $  g < \frac{1}{64} \gamma  $ and $\gamma < \frac{1}{2}$. It follows, for every $n \in \mathbb{N}_0$, that 
\begin{align}\label{Main}
\Vert P_{n+1} - \widetilde P_{n} \Vert \leq 
48 g   \rho_n \Vert R_n^\perp   \sigma_1  P_{n}  \Vert \, .
\end{align}
\end{lemma}
\noindent{\it Proof:}
By the second resolvent identity and   Cauchy's integral formula, we have that 
\begin{align}\label{cast1}
\widetilde P_{n}^\perp = & \frac{-1}{2 \pi i} \int_{\Gamma_{n+1}} (E_n -z)^{-1} - (\widetilde H_{n} -  z )^{-1} dz 
\\ \notag
= & \frac{-1}{2 \pi i} \int_{\Gamma_{n+1}} 
 (E_n-z)^{-1} (\widetilde H_{n} - E_n ) (\widetilde H_{n}  - z )^{-1}   dz \\ 
 \notag
= & \frac{-1}{2 \pi i} \int_{\Gamma_{n+1} - E_n} 
 (-z)^{-1} (\widetilde H_{n} - E_n  ) (\widetilde H_{n} - E_n - z )^{-1}   dz
 \\ \notag
= & \frac{-1}{2 \pi i} \int_{\Gamma_{n+1} - E_{n+1}} 
 (-z)^{-1} (\widetilde H_{n} - E_n  ) (\widetilde H_{n} - E_n - z )^{-1}   dz  ,
\end{align}
where we deform the contour from $ \Gamma_{n+1} - E_{n}  $ to $ \Gamma_{n+1} - E_{n+1} $ using Propositions \ref{energies} and \ref{gapgood}. Eq.~\eqref{cast1}  
implies that
\begin{align} \label{put1}
\widetilde R_{n}^{\perp} 
= & \frac{-1}{2 \pi i} \int_{\Gamma_{n+1} - E_{n+1}} 
 (-z)^{-1} (\widetilde H_{n} - E_n - z )^{-1}  dz  . 
 \notag \\
\end{align}
Similarly, we get
\begin{align}\label{put2}
R_{n+1}^{\perp}  =  \frac{-1}{2 \pi i} \int_{\Gamma_{n+1} - E_{n+1}} 
 (-z)^{-1} (H_{n+1} - E_{n+1} - z )^{-1}  dz \, . 
\end{align}
Hence, using the second resolvent identity again, we obtain

\begin{align} \label{ppa}
 R_{n+1}^\perp  -  \widetilde R_{n}^\perp  
  = &  \frac{-1}{2 \pi i} \int_{\Gamma_{n+1} - E_{n+1}}  \Big [z^{-1}  (H_{n+1} - E_{n+1} - z )^{-1} \big ( \Phi(\widetilde G_n) + E_n - E_{n+1} \big ) 
   \notag \\ &  \hspace{7.3cm}\cdot (\widetilde H_{n} - E_n - z )^{-1} \Big] dz .
\end{align}
We notice that
\begin{equation}\label{pp1}
 (H_{n+1} - E_{n+1} - z )^{-1}  =  (H_{n+1} - E_{n+1} - z )^{-1} P_{n+1}^\perp - 
 z^{-1}P_{n+1} 
\end{equation} 
and 
\begin{equation}\label{pp2}
 (\widetilde H_{n} - E_{n} - z )^{-1}  =  (\widetilde H_{n} - E_{n} - z )^{-1} \widetilde P_{n}^\perp - 
 z^{-1}\widetilde P_{n}. 
\end{equation} 
Inserting Eqs.~\eqref{pp1} and \eqref{pp2} in \eqref{ppa} and using the 
 Cauchy's integral formula for the derivative of a function, we arrive at  
 (notice that the first terms in the right hand side of Eqs.~\eqref{pp1} and \eqref{pp2} are analytic in the interior of $\Gamma_{n+1}$)
\begin{align} \label{pppp}
 R_{n+1}^\perp  -  \widetilde R_{n}^\perp \nonumber 
    = & 
        P_{n+1}  \big (\Phi(\widetilde G_{n}) - E_{n+1}  + E_n \big) 
  {\widetilde R^2_{n}} \widetilde P_{n}^\perp \nonumber \\
  & 
 +    {R^2_{n+1}} P_{n+1}^\perp   \big ( \Phi(\widetilde G_{n}) - E_{n+1}  + E_n \big ) 
  \widetilde P_{n} \nonumber \\
&  
 -  R_{n+1}^\perp \big ( \Phi(\widetilde G_{n}) - E_{n+1}  + E_n \big ) 
  \widetilde R_{n}^\perp \, .
\end{align}
Adding $  \widetilde R_{n}^\perp $ in Eq.~\eqref{pppp} and applying it to 
$ \Big ( P_{n}^\perp   \otimes P_{\widetilde \Omega_n}^\perp \Big )  \sigma_1 
 \Big ( P_{n}  \otimes P_{\widetilde \Omega_n}^\perp  \Big )  =  P_n^{\perp} \sigma_1 P_n \otimes P_{\widetilde \Omega_n }^{\perp}$ leads us to
\begin{align}
 R_{n+1}^\perp  \Big ( P_{n}^\perp   \otimes P_{\widetilde \Omega_n}^\perp \Big )  \sigma_1  \Big ( P_{n}  \otimes P_{\widetilde \Omega_n}^\perp  \Big ) 
 =& \Big[\mathds{1} +     P_{n+1}  \big ( \Phi(\widetilde G_n) - E_{n+1}  + E_n \big ) 
  {\widetilde R^\perp_{n}} \nonumber \\
  & - R_{n+1}^\perp  \big ( \Phi(\widetilde G_n) - E_{n+1}  + E_n \big ) 
  {\widetilde P^\perp_{n}}  \Big] \nonumber \\
  & \hspace{1.3cm} \cdot \widetilde R_{n}^\perp  P_n^{\perp} \sigma_1 P_n \otimes P_{\widetilde \Omega_n }^{\perp},
    \label{projdiff.3}
\end{align}
where we used that $ \widetilde P_n \Big ( P_{n}^\perp   \otimes P_{\widetilde \Omega_n}^\perp \Big )  =0  $. \\
Using Eq.~\eqref{projdiff.3}, together with Corollary \ref{fur1c},
and Propositions  \ref{energies} and \ref{gapgood}, we obtain
\begin{align}\label{tt1}
\Big \| R_{n+1}^\perp  \Big ( P_{n}^\perp &   \otimes P_{\widetilde \Omega_n}^\perp \Big )  \sigma_1  \Big ( P_{n}  \otimes P_{\widetilde \Omega_n}^\perp \Big ) \Big \|
\\ \notag  & \leq 
\Big ( 1 + (16  + 32  + 2  + 1  )  \frac{ g}{\gamma}  \Big ) \Big \|
\widetilde R_{n}^\perp  P_n^{\perp} \sigma_1 P_n \otimes P_{\widetilde \Omega_n }^{\perp} \Big \|
\\ \notag & \leq  3 \Big \|
\widetilde R_{n}^\perp  P_n^{\perp} \sigma_1 P_n \otimes P_{\widetilde \Omega_n }^{\perp} \Big \|  \leq   6\Vert R_n^\perp   \sigma_1  P_{n}  \Vert , 
\end{align}
since
\begin{align}
 \Vert \widetilde  R_{n}^\perp  P_n^{\perp} \sigma_1 P_n \otimes P_{\widetilde \Omega_n }^{\perp} \Vert \nonumber  \\
  \leq &
\big ( {1} + \Vert (H_n -E_n + H_{ph}(\widetilde \omega_n)  )^{-1} H_{ph}(\widetilde \omega_n) [P_{n}^\perp   \otimes P_{\widetilde \Omega_n}^\perp]  \Vert \big ) \nonumber \\
& \hspace{7.5cm} \cdot \Vert R_n^\perp   \sigma_1  P_{n}  \Vert \nonumber \\
  \leq &
\big ( {1} + \sup_{r \in [\rho_{n+1} , \infty) } \Vert (H_n -E_n +r )^{-1} r  \Vert \big ) \Vert R_n^\perp   \sigma_1  P_{n}  \Vert \nonumber \\
 \leq &
2 \Vert R_n^\perp   \sigma_1  P_{n}  \Vert \, , \label{projdiff.4}
\end{align}
and we conclude by putting together Eqs.~\eqref{lemma71}, \eqref{onem}, \eqref{tt1} and \eqref{projdiff.4}, which leads us to Eq.~\eqref{Main}.  
\qed

\subsection{Main Results : Convergence of the Regularized Ground State Projections and Existence of the Ground State of $H$}\label{Main-Section}
Here we prove our principal theorems. We collect all results of Section \ref{keyestimates} (Lemmas \ref{lemma.7}, \ref{Onesimo} and \ref{Onesimo2}) to prove  our first main theorem, Theorem \ref{MAIN}, which is the most difficult and technical result in the present paper.  Theorem \ref{MAIN} gives the convergence of the sequence of ground state projections $ \{  P^\infty_n \}_{n \in \mathbb{N}_0}$. The  limit of this sequence is denoted by  $ P_{\rm gs} $, its range actually consists of ground state eigenvectors corresponding to the full energy operator $H$, as it is demonstrated in Theorem \ref{GS-Projection}.

\vspace{.5cm}

\subsubsection{Convergence of the regularized Ground State Projections }

\vspace{.3cm}

\begin{theorem}\label{MAIN}
Suppose that $  g < \frac{1}{64} \gamma  $ and $\gamma < \frac{1}{2}$. It follows, for every $n \in \mathbb{N}_0$, that
\begin{align}\label{Eqmain}
\|  P_{n+1} - \widetilde P_{n}  \| \leq 48 g   \Big (  \gamma + 147  g  \Big )^n  . 
\end{align}
In particular, if additionally $\gamma \leq 1/8$, then 
\begin{equation}\label{convv}
\|  P_{n+1} - \widetilde P_{n}  \| \leq \Big ( \frac{1}{2}\Big )^{n},
\end{equation}
and the sequence $(P_n^\infty)_{n \in \mathbb{N}_0}$, see \eqref{extra}, converges to a rank-one projection
\begin{equation}\label{convv1}
P_{\rm gs} = \lim_{n \to \infty } P_n^{\infty}. 
\end{equation}  
\end{theorem}
\noindent{\it Proof: }
First we notice that 
\begin{align} \label{notice}
R_n^\perp   \sigma_1  P_{n} = R_n^\perp   \sigma_1 \widetilde P_{n-1} + R_n^\perp   \sigma_1 [ P_{n} - \widetilde P_{n-1} ]. 
\end{align}
Next we consider the following computations, which use Eq.~\eqref{pppp} with $n-1$ replacing $n$, 
\begin{align}
 R_{n}^\perp  \sigma_1  \widetilde P_{ n -1}  
 =& \Big[{1} +     P_{n}  ( \Phi( \widetilde G_{n-1}) - E_{n}  + E_{n-1} ) 
  {\widetilde R^\perp_{{n-1}}} \nonumber \\
  & \hspace{3cm}- R_{n}^\perp  ( \widetilde \Phi(G_{n-1}) - E_{n}  + E_{n-1} ) 
  {\widetilde P^\perp_{{n-1}}}  \Big] \nonumber \\
  & \hspace{6cm}\cdot  \widetilde R_{{n-1}}^\perp  [ P_{{n-1}}^\perp   \otimes P_{\widetilde \Omega_{n-1}}]  \sigma_1   \widetilde P_{{n-1}}  \nonumber \\
   =& \Big[{1} +     P_{n}  ( \widetilde \Phi(G_{n-1}) - E_{n}  + E_{n-1} ) 
  {\widetilde R^\perp_{{n-1}}} \nonumber \\
  & \hspace{3cm} - R_{n}^\perp  ( \widetilde \Phi(G_{n-1}) - E_{n}  + E_{n-1} ) 
  {\widetilde P^\perp_{{n-1}}}  \Big] \nonumber \\
  & \hspace{6cm} \cdot [ R_{{n-1}}^\perp  \sigma_1  P_{{n-1} } ] \otimes  P_{\widetilde \Omega_{n-1}} \, . \label{projdiff.5}
\end{align}
Here we use the key symmetry
\begin{align}
\widetilde P_{n-1} \sigma_1 \widetilde P_{n-1} = 0,   
\end{align}
see Lemma \ref{key}, to drop the contribution from the second line in 
Eq.~\eqref{pppp}. We also use $\widetilde P_{n-1}^{\perp} = P_{n-1}^{\perp} +  P_{n-1} \otimes P_{\widetilde \Omega_{n-1}}^{\perp}$ to prove that 
$$ 
 \widetilde P_{n-1}^{\perp} \sigma_1 \widetilde P_{n-1}  =  P_{n-1}^{\perp} \otimes  P_{\widetilde \Omega_{n-1}} \sigma_1 \widetilde P_{n-1} . 
 $$   
With the help of  Eq.~\eqref{projdiff.5}, together with Corollary \ref{fur1c}
and Propositions \ref{energies} and \ref{gapgood}, we obtain
\begin{equation}\label{mmm}
 \| R_{n}^\perp  \sigma_1  \widetilde P_{ n -1} \| \leq 
  \Big ( 1 + 51    \frac{g}{\gamma} \Big )  \cdot \| R_{{n-1}}^\perp  \sigma_1  P_{{n-1} } \| , 
\end{equation}
where we argue as in Eq.~\eqref{tt1}. Lemma \ref{Onesimo2}, together with 
Eqs.~\eqref{notice} and \eqref{mmm} (see also Proposition \ref{gapgood}), imply (notice that $\gamma < 1/2$)
\begin{align}\label{mmm1}
\| R_n^\perp   \sigma_1  P_{n}  \| \leq   \Big (  1 + (96 + 51) \frac{g}{\gamma}  \Big ) \| R_{{n-1}}^\perp  \sigma_1  P_{{n-1} } \|. 
\end{align} 
Then we inductively get, using Lemma \ref{Onesimo2} and \eqref{mmm1}, that
\begin{align}
\| P_{n+1} - \widetilde P_n \|   & \leq 48 g   \rho_n  \Big (  1 + 147  \frac{g}{\gamma}  \Big )^n  \| R_{{0}}^\perp  \sigma_1  P_{{0} } \| \\ \notag & \leq
 48 g   \Big (  \gamma + 147  g  \Big )^n ,
\end{align}
where we use \eqref{gap0}, recalling that $\kappa <1$. This establishes Eq.~\eqref{Eqmain}. Eqs. \eqref{convv} and \eqref{convv1} are direct consequences of \eqref{Eqmain}. Clearly $P_{\rm gs}$ being limit of rank-one projections (see Corollary \ref{rankone}) is rank-one. 

\qed

\vspace{.5cm}

\subsubsection{Construction of the Ground State Projection of $H$}

\vspace{.3cm}

\begin{theorem}\label{GS-Projection}
Suppose that $  g < \frac{1}{64}  \gamma $ and $\gamma < \frac{1}{8}$. Then the range of $P_{\rm gs}$ is contained in the domain of $H$ and
\begin{equation}\label{nopuse}
H P_{\rm gs}  = E_{\rm gs}  P_{\rm gs},
\end{equation}
where $   E_{\rm gs}   = \lim_{n \to \infty }  E_n $. 
\end{theorem}
\begin{proof}
We denote  
\begin{align}\label{widetildesinfty}
 \omega_n^{\infty} (k) := &  {\bf 1}_ { B_n } \, \omega \, , &   G_n^{\infty} (k) := &    {\bf 1}_ {  B_n }  \, G \,  ,
\end{align}
and define $ H_\rad(\omega_n^{\infty}) $ and $\Phi(   G_n^{\infty} )$ as in 
Eqs.~\eqref{eq.3} and \eqref{eq.7} on $\mathcal{F}_n^{\infty}$ (see \eqref{eq.2.2}). 
Since $H_{ph}(\omega^n_\infty) P_n^{\infty}  =0$, see \eqref{extra}, we have that
\begin{align}
H P_n^{\infty} 
= & E_{n} P_n^{\infty}+   \Phi(G^n_\infty)  P_n^{\infty} \, . \label{very.last.eq}
\end{align}
As 
\begin{equation}\label{q1}
\lim_{n \to \infty } E_n P^{\infty}_n = E_{\rm gs} P_{\rm gs}
\end{equation}
and 
\begin{equation}\label{q2}
\lim_{n \to \infty} \| \Phi(G^n_\infty)  P_n^{\infty} \|   = \lim_{n \to \infty } \| G^n_\infty  \|_{L^2} = 0, 
\end{equation}
we obtain that 
\begin{equation}\label{q3}
\lim_{n \to \infty}  H P_n^{\infty}  =  E_{\rm gs} P_{\rm gs}. 
\end{equation}
Since the sequence of projections $(P_n^{\infty})_{n \in \mathbb{N}}$ converges, we can find $ N \in \mathbb{N} $ and a vector $\phi \in \mathcal{H}$ such that $ \phi_n : = P_n^{\infty} \phi  \ne 0$, for all $n \geq N$, and $ \psi : =  P_{\rm gs} \phi \ne 0$. Then we have 
\begin{equation}
\psi = \lim_{n \to \infty} \phi_n, \hspace{2cm} E_{\rm gs} \psi = \lim_{n \to \infty} H \phi_n, 
\end{equation}
 where we use  Eq.~\eqref{q3}. As $H$ is a closed operator, $\psi $ belongs to its domain and 
\begin{align}\label{ahorasi}
H \psi =E_{\rm gs} \psi.
\end{align}
The fact that $P_{\rm gs}$ is rank-one and Eq.~\eqref{ahorasi} imply that the range of $  P_{\rm gs} $ is contained in the domain of $H$ and Eq.~\eqref{nopuse}.

\end{proof}

\begin{remark} \label{Remark}
It is not difficult to prove that $E_{\rm gs } = \lim_{n \to \infty} E_n$ is actually the infimum of the spectrum of $H$. In fact, Lemma \ref{nuevos} implies that
$$ 
 {\rm inf} \, \sigma (H) + \rho_n   \geq E_n +  (1 -    g  )
    \rho_n ,  
$$
 for all $n$. Therefore $  {\rm inf} \,  \sigma (H) \geq E_{\rm gs} $. As $E_{\gs}$ is itself a spectral point of $H$, it equals $\inf \sigma(H)$.   
\end{remark}


\end{document}